\documentclass[submission,copyright]{eptcs}
\providecommand{\event}{ACT 2022} % Name of the event you are submitting to
\usepackage{hyperref}
\usepackage{graphicx}
\usepackage[utf8]{inputenc}
\usepackage[T1]{fontenc} 
\usepackage{amsmath,amssymb,amsfonts}
\usepackage{comment}
\usepackage{algorithm}
\usepackage{algorithmic}
\usepackage[T1]{fontenc}
\usepackage{mdframed}
\usepackage[all,cmtip]{xy}
\usepackage{blkarray,kbordermatrix}
\usepackage{url}
\usepackage{lineno}
\usepackage{xcolor}
\usepackage{soul}
\newcommand{\msccat}{\textbf{Msc}}
\newcommand{\msccatm}{\textbf{Mscm}}
\newcommand{\msicat}{\textbf{Msin}}

\usepackage{amsthm}
\theoremstyle{definition}
\newtheorem{definition}{Definition}[section]

\theoremstyle{theorem}
\newtheorem{proposition}{Proposition}[section]

\title{On a federated architecture for utilizing domain-specific descriptive models\thanks{This work was carried out and funded in the framework of the Labex MS2T. It was supported by the French Government, through the program "Investments for the future" managed by the National Agency for Research (Reference ANR-11-IDEX-0004-02). This work was co-funded by the Airbus Defence and Space company.}}
\author{Freddy Kamdem Simo
	\institute{$^{1,}$ $^2$}
	\email{frks@protonmail.ch}
	\and
	Dominique Ernadote 
	\institute{$^2$Airbus Defence and Space\\
		Elancourt, France}
	\email{dominique.ernadote@airbus.com}
	\and
	Dominique Lenne
	\institute{$^1$Sorbonne Universit\'es, Universit\'e de Technologie de Compi\`egne}
	\institute{Heudiasyc, CNRS, Centre de recherche Royallieu, CS 60319\\
		Compi\`egne, France}
	\email{dlenne@utc.fr}
}
\def\titlerunning{On a federated architecture for utilizing domain-specific descriptive models}
\def\authorrunning{F. Kamdem Simo, D. Ernadote \& D. Lenne}
\begin{document}
	\maketitle
	
	\begin{abstract}
		In a Systems Engineering setting, various models are produced using a variety of methods and tools. Focusing on a type of models -- called descriptive models -- which we shall describe, we argue that, while the clarity and precision of models are essential for their exchange and reuse, the way in which the data of these models are defined and the information conveyed is also essential for their (re) utilization -- e.g., for analysis or synthesis purposes. Category Theory has made it possible to link seemingly separate fields or domains, so that anything that can be rewritten in this framework benefits from a level of abstraction and a relational viewpoint -- essential to address complexity. We therefore take advantage of this framework to define a federated architecture for projecting and conveying these models without sacrificing clarity and precision.
		A federated architecture has two important advantages. On the one hand, it unifies these models from their structure; on the other hand, it allows a specific usage (instantiation or interpretation) of this structure within a business domain.  We define the structure of these models as a symmetric multicategory. In particular, we rely on matrices over a semiring to define morphisms and their composition. The choice of matrices is intended to facilitate the application in practice. %Following these definitions, we sketch computational and data structures, base algorithms and data format to manipulate the mathematical objects i.e. to take advantage of FA in practice.
	\end{abstract}

\section{Challenges with metamodels, data formats, interfaces and  mappings}
\label{sec:problemchap4} 
In a Systems Engineering setting, a model of a System-Of-Interest often refers to a partial, but necessary, view or aspect of that system, in pursuit of a goal. This model is not always based on a theory and need not be in many cases, e.g., in creative, preliminary or general communication phases. However, when it is necessary to reason or compute on a model, the formal content (sometimes called semantics or business semantics) of that model should (mathematically) be defined  in order to ensure sound reasoning. We shall see that such a definition is not enough to ensure that different implementations in modelling tools are equivalent or, ideally, functionally identical.
In this paper, we focus on descriptive models -- hereafter referred to simply as models, unless otherwise specified -- by which we mean a model whose structure -- equipped with input and output interaction points -- encodes the main (abstract) structure of the thing to be modelled. We examine how descriptive models produced using a particular tool might be re-utilized, for instance for analysis or synthesis purposes. The question is to what extent descriptive models can be decoupled from the logic of their implementation by a particular tool while retaining their full meaning and being reusable.

One solution for making models transferable and interoperable between some tools is to impose the following:

i) a common metamodel, ontology, schema or modelling language;

ii) a common data format;

iii) a common interface.

Metamodels support, through the definition of syntax, the graphical representation of models and descriptions of their semantics. However, when models are used outside the tool in which they originated, some assumptions present in the metamodels may be missing. For instance, Statecharts \cite{harel1987statecharts}, a visual modelling technique, have a plethora of existing semantics implemented within tools, falling into three main categories \cite{eshuis2009reconciling}.
Regarding the implementation of semantics by modelling tools, a missing assumption relating to a semantics will either be implicit in the metamodel, or explicit but proper to the modelling tool. It has been argued in \cite{seshia2014modeling} that even for the same modelling formalism, that is to say a mathematical object, different implementations (via an abstract syntax and a semantics) may exist, and so the implementation of the same semantics is not guaranteed. In addition it has been pointed out in \cite{diskin2014category} that it is hard to mathematically classify the numerous modelling languages and techniques, and that a number of MDE (Model-Driven Engineering) approaches are not grounded on formal semantics.

As for the data format, it is intrinsically related to the structure of data of models. Using a particular data format is useful when persisting models. The data format does not influence the semantics of a model.

An interface can be useful for manipulating models, but the question of how to define this interface remains. As with the above-mentioned problems with metamodels, it is necessary to ensure that the internal implementation associating to an interface reflects the modelling formalism. Indeed, when several implementations coexist, it is not certain that they behave in the same way. In the latter case, the only solution might be to use the same implementation -- like software libraries.

One might also think of mappings, translations or transformations between different metamodels, formats and interfaces. When feasible, this would imply the same problems -- managing of mappings and even new metamodels, formats or interfaces acting as bridges -- as with models. In some cases, although the magnitude of mappings or translations can be quadratic or linear (in the number of metamodels, for instance), it may be impossible (different tools specializing in different fields) or inappropriate (e.g. too expensive) to create or define a mapping. In order to address these various issues, this paper introduces and specifies a novel Federated Architecture (FA) for handling descriptive models.

In Section \ref{sec:approachchap4} we present our main idea, as well as a description of FA, and give some reasons (others are presented in Section \ref{sec:relatedresearchchap4}) for adopting the category theory framework. In Section \ref{sec:applicationsct} we then give a formal definition of FA. In Section \ref{sec:relatedresearchchap4} we discuss the relevant related work. In Section \ref{sec:conclusionchap4}, we present some concluding remarks.

\section{Main idea: encapsulation-differentiation}
\label{sec:approachchap4} 
It turns out from Section \ref{sec:problemchap4},  that whatever the common interface, metamodel or data format used between tools,  clarity and precision are not the only essential features. In today's software-intensive and computer-aided environments, the way in which models can be defined and conveyed is also essential since this has a direct impact on interoperability. We consider this a question that is better addressed at the level of the architecture of models. We adopt the ISO/IEC/IEEE 42010\footnote{\url{http://ieeexplore.ieee.org/document/6129467/}} standard, which defines the architecture of any system as the abstract description or the fundamental organization of that system, embodied in its components, the relations of those components to each other and to the environment, and the principles governing the design and evolution of that system.

FA is federated in that while it is based on a specific structure, it makes it possible to define domain specific-models via instantiations and interpretations -- usages -- of that structure. FA therefore unifies models on their structure and the links from the structure to its usages. FA equally differentiates models through  domain-specific interpretations. Below we look at the basic organisation and principles governing the definition of a descriptive model.

\textbf{Structure of a model}
\label{sec:modelorganisationstructure} 
Above we define a descriptive model as one whose structure -- equipped with input and output interaction points -- encodes the main (abstract) structure of the thing to be modelled or designed. This structure consists of boxes linked with incoming and outgoing wires. Input and output interaction points are ports associated to boxes ensuring connection with their surrounding environment.
A representation of a tricky system can be managed using hierarchical decompositions and re-compositions.

\textbf{Interpretations of the structure}
Since many systems are usually studied from different perspectives, namely: human, physical, structure, behaviour, cost, safety, etc., then different interpretations of their structure(s) are possible and necessary. At the level of models, these interpretations yield domain-specific languages or semantics. A structure hence constrains the definition of its interpretation, even if possible interpretations of the structure cannot be known by examining it per se.
It may consequently be difficult, if not impossible, to impose or foresee all possible interpretations of the structure of a model. What is needed is a means of using that structure appropriately, and this starts with a clear and precise definition of that structure.

\textbf{Instances of the structure}
Utilizing a structure involves data relating to an actual (or candidate) system that is modelled. These data should also relate to an interpretation of the structure, meaning that their precise meaning can be known and that they can be utilized.

The main components of FA are therefore (a) the model's structure; (b) the interpretations of the structure; (c) the corresponding instances of the structure. %Data structures and format will also be necessary in practice for operating with these components.
The architecture needs to support the ability (or to provide means) to define various interpretations of boxes. It does not model a concrete interpretation of them, but we will present some principles or guidelines for creating an interpretation. The same remark applies to instances of the structure.

\textbf{Trivial example} suppose the system is a research paper like the present one. A Structure of the paper comprises 8 boxes (that is to say, sections, including Abstract, References and Appendix). The paper itself is seen as a composite section that embodies these 8 sections. An Instance of this structure can be given by providing each section's title, size, or content. An interpretation of this structure might be a definition of what the size of (constituent and composite) boxes  means. Note that this example does not consider the internal wiring pattern of composite boxes (sections).

\textbf{Why the category theory framework?}
 Category theory was created to unify and simplify mathematical systems \cite{Saunders1998Categories} and has given rise to some universal constructions based on a relational viewpoint. 
This framework enables us to i) naturally define the notion of composite-box with its constituent-boxes via an arrow in a category; ii) abstract away details relating to a particular interpretation of the structure of models. At the same time, it allows a formal link to be specified between the structure of models and various targets like Instances and Interpretations. This formal link is ideally defined as a functor. We discuss further in Section \ref{sec:r2}, comparing and drawing on related work, why category theory seems to be the framework for appropriately defining FA.
Besides the framework can guide us towards possible canonical extensions of FA  with the enrichment of objects and category structure, see Section \ref{sec:conclusionchap4}.

	\section{Components of FA} 
	\label{sec:applicationsct}
	
	We define the structure of models in Section \ref{sec:modelstructure}, and discuss the interpretation and usage of this structure in Section \ref{sec:specs} and Section \ref{sec:structureusage} respectively. To be able to define the structure of a model as a category, the objects will be the boxes and the arrows will show how boxes are linked to create composite boxes.
	
	\subsection{Structure of models}
	\label{sec:modelstructure}
	\begin{figure}[!ht]
		\begin{minipage}{0.5\textwidth}
		\centering
		\includegraphics[scale=0.25]{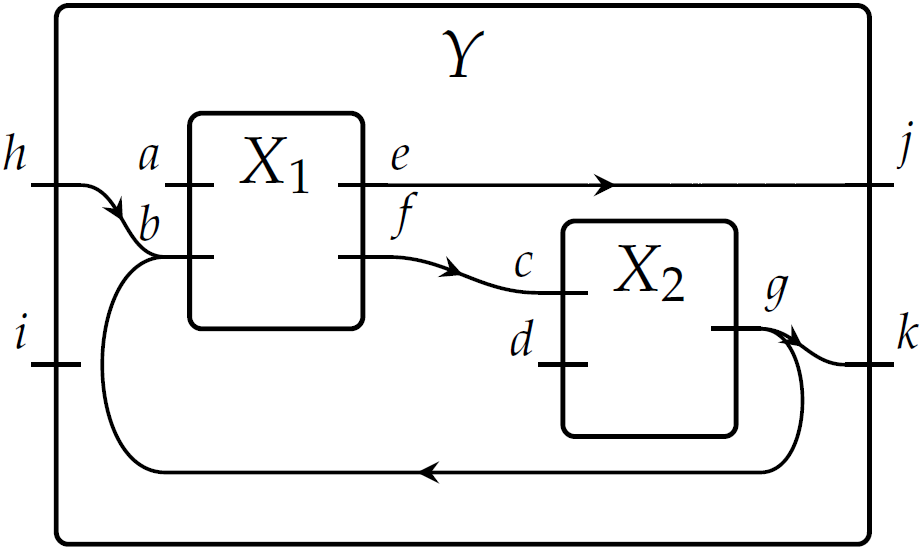}
		\caption{A box Y composed of 2 boxes X$_1$ and X$_2$\label{fig:wdbox}}
		 \end{minipage}\hfill
		 \begin{minipage}{0.46\textwidth}
		\centering
		\includegraphics[scale=0.27]{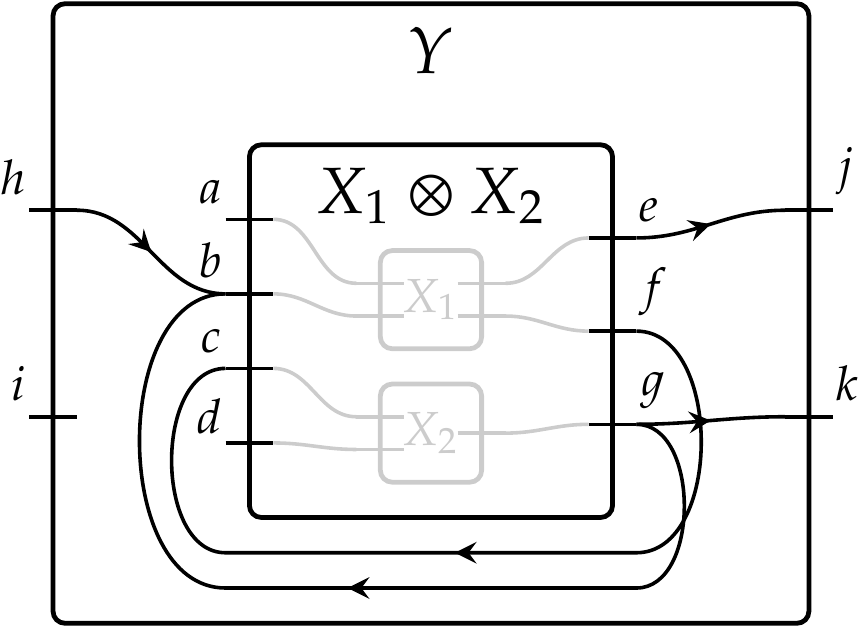}
		\caption{A box $Y$ composed of $X_1 \otimes X_2$ \label{fig:wdboxtensor}}
		\end{minipage}\hfill
	\end{figure}
	
	Defining the structure of models as a symmetric multicategory denoted \msccat\: is done in two stages.
	We start by defining a symmetric monoidal category (SMC) denoted (\msccatm, $\otimes$, I) and prove that it is indeed an SMC. We then derive \msccat\: using results from \cite{leinster2003higher}.
	Since we will be using matrices to define the composition in \msccatm, we may be able to derive a further symmetric multicategory denoted \textbf{MatMsc}, where arrows are given by matrices that might be useful for computations.
	
	Why is it not enough simply to derive an SMC? Consider the box $Y$ in Figure \ref{fig:wdbox}, composed of 2 boxes $X_1$ and $X_2$ . In an SMC $Y$ can be seen to be composed of $X_1 \otimes X_2$ (i.e. there is an arrow $X_1 \otimes X_2 \to Y$), yielding Figure \ref{fig:wdboxtensor}. Figure \ref{fig:wdbox}
	is a lot clearer than Figure \ref{fig:wdboxtensor}, in which the modularity (i.e. the clear distinction between $X_1$ and $X_2$ in $Y$) is not reflected graphically. Figure \ref{fig:wdbox} corresponds to an arrow $X_1, X_2 \to Y$ in the multicategory. However, it is more convenient to use the SMC than the multicategory insofar as subscripts can be avoided.
	
	The box and arrow representation that we use should not be confused with string diagrams
	\cite{selinger2010survey}. String diagrams are a graphical language of categories \cite{selinger2010survey} in which the objects and morphisms of categories are represented by wires and boxes respectively. In our paper, on the other hand, objects correspond to boxes, and morphisms correspond to how composite boxes are internally built (using wiring patterns) with its constituent boxes.

	Figure \ref{fig:wdbox} looks like a Wiring Diagram WD \cite{spivak2013operad} \cite{vagner2015algebras}, but in fact it differs from a WD in the following respects. In figure \ref{fig:wdbox} we allow unconnected ports (ports without incoming or outgoing wires), converging wires (separate wires with the same target port), and diverging wires (same source port). Unconnected ports and converging wires are prohibited in a WD. Figure \ref{fig:wdbox} would also allow several links between a source port and a target port. WDs may therefore be considered as special cases of figure \ref{fig:wdbox} that include certain restrictions; this is discussed in Section \ref{sec:r2}. The point of allowing these various things is to encompass as many wiring patterns as possible that may occur in practice.  For instance, it might be the case, structurally, that an actual component is not plugged into its surrounding environment via any of its interfaces.
	See for instance a simple Ptolemy model \cite[Figure 9]{tripakis2013modular}, similarly permissive.

	\begin{figure}[!h]
		\centering
		\includegraphics[scale=0.35]{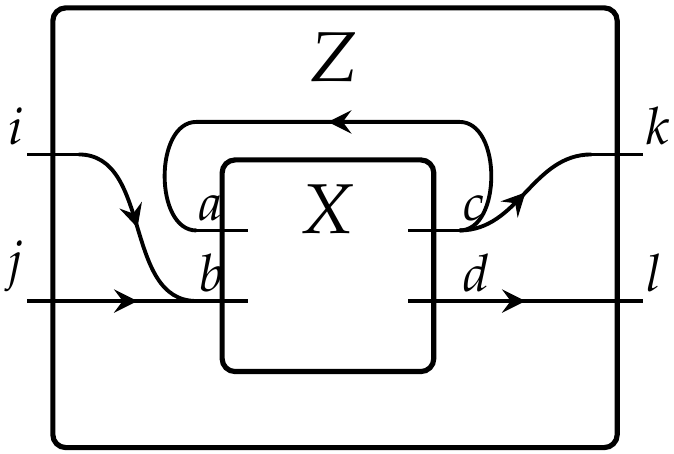}
		\caption{A box $Z$ composed of $X$ \label{fig:wdboxxz}}
	\end{figure}
	
	\begin{definition}
		\label{def:defmscm}
		The category \msccatm\:has as constituents 
	\end{definition}
	
	$\bullet$ Objects $\msccatm_0$. An object $a \in \msccatm_0$ is a box and consists of a tuple $(in(a), out(a))$, where $in(a)$ and $out(a)$ are the finite sets of input (left) and output (right) ports of (the box) $a$ such that  $in(a)\cap  out(a) = \emptyset$. %For instance, the object $Y$ is given by (\{$h,i$\}\{$j,k$\}).$
	We additionally require that \textbf{R1}: the set of ports (i.e. $in(a) \cup out(a)$) be pairwise disjoint for all objects in $\msccatm_0$.
	
	$\bullet$ Arrows  \msccatm$(b;a)$. An arrow $b \xrightarrow{\theta} a$ 
	indicates how a box $a$ is internally built from the box $b$.
	It consists of a tuple $(\theta^{in},\theta^{out})$ 
	\begin{equation}
	\begin{array}{rrrll}
	\theta^{in}: & L^{in} \rightarrow &  in(b)  &		\times &  (out(b) \cup in(a))   \\
	\theta^{out}: & L^{out} \rightarrow &     out(a)                   &		\times &  out(b)
	\end{array}
	\label{eq:arrowinmsc1}
	\end{equation} where 
	$L^{in}$ and $L^{out}$ are the abstract finite sets of links coming respectively into an input port and an output port of one of the boxes $b$ and $a$.
	
	We require that \textbf{R2}: $L^{in}, L^{out} \in K$ and  $L^{in} \cap L^{out} = \emptyset$. $(K,\cup,\cdot,\emptyset,\{\epsilon\})$ is a semiring introduced below to define composition using matrices and detailed in Appendix (Section \ref{app:conmormat}).
	We additionally require that \textbf{R3}: the domains (or more precisely $L^{in} \cup L^{out}$) of the set of morphisms be pairwise disjoint.
	R1, R2 and R3 ensure that objects and morphisms are all unambiguously defined. This will also be important in the definition of the monoidal product.
	Note also that (\ref{eq:arrowinmsc1}) does not allow a link to go from an output port of $a$ to one of its input ports.
	This is a modelling choice consistent with the fact that any box only exposes its input and output ports. An arrow is intended to give the internal wiring pattern of a composite box.
	A link that goes from an output port of $a$ to one of its input ports should, rather, give rise to a new object $a'$ and a new arrow  $a \to a'$ which models a kind of self-feedback, see Figure \ref{fig:validinvalidbox}.
	\begin{figure}[!h]
		\centering
		\includegraphics[scale=0.45]{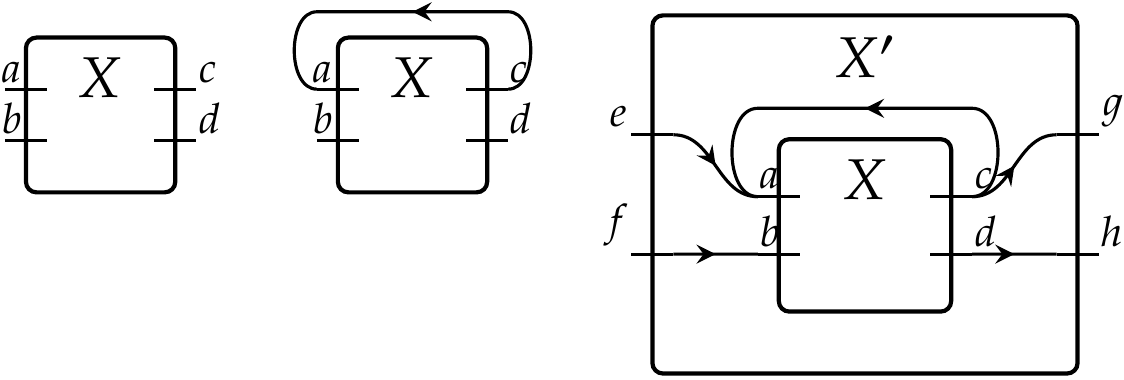}
		\caption{From left to right: a box X, an invalid-box X and an arrow $X \to X'$ \label{fig:validinvalidbox}}
	\end{figure}
	
	In the same way, (\ref{eq:arrowinmsc1}) does not allow a link to go directly from an input port of $a$ to an output port.
	This is also a modelling choice that allows us to consistently zoom in on a composite box, avoiding the presence of isolated links when applying composition. If there were an isolated link, we would simply need to create a basic box with one input port and one output port and link them to the input and output ports respectively of $a$, such that this box outputs onto its output port the data that comes into its input port.
	
	\textbf{Example} Consider Figure \ref{fig:wdboxxz}. We have two boxes X and Z given by $(\{a,b\},\{c,d\})$ and $(\{i,j\},\{k,l\})$ respectively. There is an arrow $X \to Z$ given by
		(\{$l_1 \mapsto (a,c), l_2 \mapsto (b,i), l_3 \mapsto (b,j)$\},\{$l_4 \mapsto (k,c) , l_5 \mapsto (l,d) $\}). 
	
	$\bullet$ Identities \msccatm$(a;a)$. An identity $a \xrightarrow{1_a} a$ consists of  $\theta^{in}$ and $\theta^{out}$
	given by\\
	$(1^{in}_a,1^{out}_a) := (\{lin_i \mapsto (i,i), i \in in(a)\}, \{lout_j \mapsto (j,j), j \in out(a)\})$ ; i.e. links connect identical (input to input and output to output respectively) ports from the domain ($a$) to the codomain ($a$) of $1_a$.
	
	$\bullet$ Composition formula $\circ$. Composition enables a constituent of a composite box to be replaced by a more detailed constituent(s). For boxes $a, b, c \in$ \msccatm$_0$, we have a function $\circ$:
	\begin{center}$\msccatm(b; a) \times \msccatm(c ; b)  \rightarrow \msccatm(c ; a)$.\end{center}
	We must define an arrow $c \xrightarrow{\theta \circ \theta_1} a$. This amounts to finding a function $\theta_0 :=  \theta \circ \theta_1$, i.e. $\theta^{in}_0$ and $\theta^{out}_0$ following (\ref{eq:arrowinmsc1}).
	
	In order to define the composition formula, we use matrix operations. We start by establishing a connection between arrows and certain matrices over a semiring $(K,\cup,\cdot,\emptyset,\{\epsilon\})$; see Appendix \ref{app:conmormat} for more details. We are now ready to work with the  matrix representation of an arrow.
	
	The data (\ref{eq:arrowinmsc1}) of an arrow can be rewritten as follows (see Appendix \ref{app:conmormat}).
	We define
	\begin{equation}
	\begin{array}{rrrl}
	M^{out} :& Y^{in} \times (Y^{out} \cup Z^{in})& \to& 2^{L^{in}}\\
	M^{in} :& Z^{out} \times Y^{out}& \to& 2^{L^{out}}
	\end{array}
	%\vspace{-0.1cm}
	\label{eq:arrowinmsc4} 
	\end{equation}\\
	where $Y^{in} = in(b)$, $Y^{out} = out(b)$,  $Z^{in} = in(a)$ and $Z^{out} = out(a)$. $2^L$ is the powerset.
	
	Let $X^{in} = in(c)$, $X^{out} = out(c)$.
	
	We denote $\emptyset_M$ as the matrix with $\emptyset$'s everywhere and $\mathbb{I}_M$ the square matrix with $\{\epsilon\}$'s on the main diagonal and $\emptyset$'s elsewhere.

	Given (\ref{eq:arrowinmsc4}) in Appendix \ref{app:conmormat} and the above notations, we are now ready to define arrow composition.

	The map $\theta \circ \theta_1$ (i.e., $c \xrightarrow{\theta_0}a$ in \msccatm)  of an arrow is rewritten as $(O^{in},O^{out})$, given by the dashed arrows in the following diagrams. In accordance with Appendix \ref{app:conmormat} an arrow $I \xrightarrow{M} J$ in this diagram defines the matrix $M(I,J)$.\\ %\st{where $\leq_S$ is a permutation of $S$, (we recall that the elements of S will be the ports of boxes). For convenience, we abuse the notation and  identify $\leq_S$ to a permutation of elements of $S$}.\\
	$\begin{array}{lr}
	\xymatrixcolsep{4pc}\xymatrix {
		{X^{in} \ar[d]_-{N_{in}}} \ar@{-->}[r]^-{O^{in}}  &{Z^{in} \cup X^{out}}\\
		{Y^{in} \cup X^{out}}  \ar[r]_-{M'^{in}}             &{Z^{in} \cup Y^{out} \cup X^{out}} \ar[u]_-{  N'^{out} } 
	}\quad\quad
	&
	\xymatrixcolsep{3pc}\xymatrix {
		{Z^{out}} \ar@{-->}[r]^-{O^{out}} \ar[d]^-{M^{out}} &{X^{out}}\\
		{Y^{out}} \ar[ru]_-{N^{out}}
	}
	\label{eq:compostionformula2}
	\end{array}$\\
	%+ 1_{X^{out}}
	%+1_{Z^{in}}+ 1_{X^{out}}
	\begin{equation}
	\label{matrixarrowcomposition}
	\begin{array}{rl}
	O^{out} =& M^{out} \times N^{out}\\
	O^{in} = & N^{in} \times M'^{in} \times N'^{out}
	\end{array}
	%\label{eq:arrowinmsc4} 
	\end{equation} where
	%$N^{in}(X^{in}\times (Y^{in} \cup X^{out})$, $M'^{in}( (Y^{in} \cup X^{out}) \times (Z^{in} \cup Y^{out} \cup X^{out}) )$ and $N'^{out}((Z^{in} \cup Y^{out} \cup X^{out}) \times (Z^{in}  \cup X^{out}) )$ are respectively given by:\\
	$N^{in}$, $N'^{out}$, $M'^{in}$ are respectively given by:
	\[
	%\begin{tabular}{c c c}
	%	$M^{out}$ & $N^{out}$ & $N^{in}$ \\
	\begin{blockarray}{rl}
	& Y^{in}\cup X^{out}  \\
	\begin{block}{r(c)}
	X^{in} & N^{in} \\
	\end{block}
	\end{blockarray}
	\quad
	\begin{blockarray}{rcc}
	& Z^{in} & X^{out} \\
	\begin{block}{c(cc)}
	Z^{in}  &  \mathbb{I}_M &     \emptyset_M \\
	Y^{out} & \emptyset_M   &  N^{out}  \\
	X^{out} & \emptyset_M   &     \mathbb{I}_M \\
	\end{block}
	\end{blockarray}
	\quad
	\begin{blockarray}{rcc}
	& Z^{in}\cup Y^{out}  & X^{out} \\
	\begin{block}{c(cc)}
	Y^{in}  & M^{in}        &     \emptyset_M \\
	X^{out} & \emptyset_M   &  \mathbb{I}_M  \\
	\end{block}
	\end{blockarray}
	\]
	
	and $M^{out}$, $N^{out}$ are respectively given by:\\
	\[
	%\begin{tabular}{c c c}
	%	$M^{out}$ & $N^{out}$ & $N^{in}$ \\
	\begin{blockarray}{rl}
	&Y^{out}  \\
	\begin{block}{r(c)}
	Z^{out} & M^{out} \\
	\end{block}
	\end{blockarray}
	\quad
	\begin{blockarray}{rl}
	&X^{out}  \\
	\begin{block}{r(c)}
	Y^{out} & N^{out} \\
	\end{block}
	\end{blockarray}
	\]
	\begin{figure}[!h]
		\centering
		\vspace{-0.2cm}
		\includegraphics[scale=0.25]{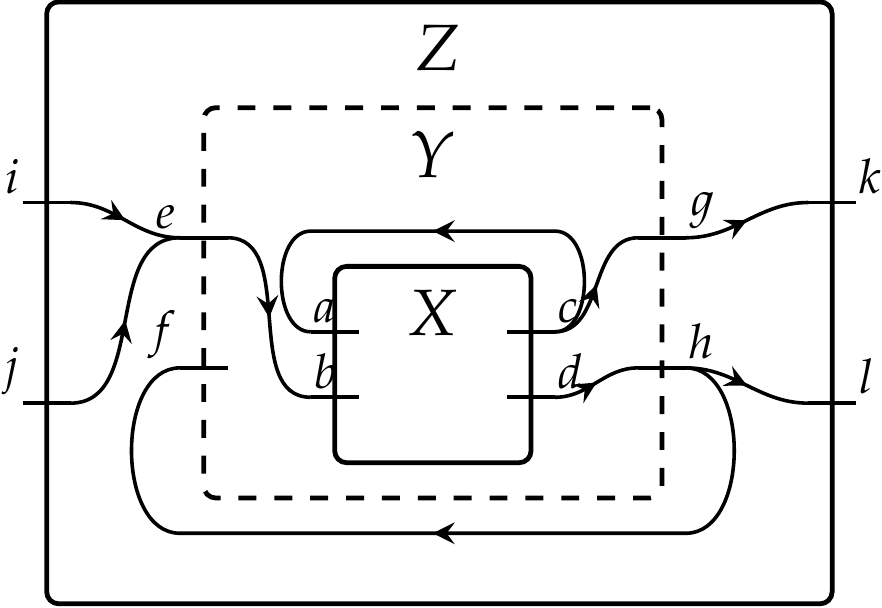}
		\caption{A box $Z$ composed of $Y$, itself composed of $X$ \label{fig:wdboxxyz}}
	\end{figure}
	
	\textbf{Example} Figure \ref{fig:wdboxxyz} is the representation of two arrows $X \xrightarrow{\theta_1} Y$ and $Y \xrightarrow{\theta} Z$. The arrow $X \xrightarrow{\theta  \circ \theta_1}  Z$  is given by the matrices $O^{in}$ and $O^{out}$ computed as follows.
		The matrices $M^{out}$, $N^{out}$, $N^{in}$, $N'^{out}$ and $M'^{in}$  are the following.
		\[
		%\begin{tabular}{c c c}
		%	$M^{out}$ & $N^{out}$ & $N^{in}$ \\
		M^{out}=\begin{blockarray}{rcc}
		& g & h \\
		\begin{block}{c(cc)}
		k & \{l_1\} & \emptyset  \\
		l & \emptyset & \{l_2\} \\
		\end{block}
		\end{blockarray}
		\quad
		%&
		N^{out}=\begin{blockarray}{rcc}
		& c & d \\
		\begin{block}{c(cc)}
		g & \{l_3\} & \emptyset  \\
		h & \emptyset & \{l_4\} \\
		\end{block}
		\end{blockarray}
		\quad
		%&
		N^{in}=\begin{blockarray}{rcccc}
		& e & f & c & d\\
		\begin{block}{c(cccc)}
		a & \emptyset & \emptyset & \{l_5\} & \emptyset  \\
		b & \{l_6\} & \emptyset & \emptyset & \emptyset  \\
		\end{block}
		\end{blockarray}
		%\end{tabular}
		\]
		%
		%\\
		\[	
		%\begin{tabular}{c  c}
		%	$N'^{out}$ & $M'^{in}$ \\
		N'^{out}=\begin{blockarray}{ccccc}
		& i & j & c & d\\
		\begin{block}{c(cccc)}
		i & \{\epsilon\} & \emptyset & \emptyset & \emptyset \\
		j & \emptyset & \{\epsilon\} & \emptyset & \emptyset \\
		g & \emptyset & \emptyset & \{l_3\} & \emptyset \\
		h & \emptyset & \emptyset & \emptyset & \{l_4\} \\
		c & \emptyset & \emptyset & \{\epsilon\} & \emptyset \\
		d & \emptyset & \emptyset & \emptyset & \{\epsilon\} \\
		\end{block}
		\end{blockarray}
		\quad
		%&
		M'^{in}=\begin{blockarray}{rcccccc}
		& i & j & g & h & c & d \\
		\begin{block}{r(cccccc)}
		e & \{l_7\} & \{l_8\} & \emptyset & \emptyset & \emptyset & \emptyset  \\
		f & \emptyset & \emptyset & \emptyset & \{l_9\} & \emptyset & \emptyset \\
		c & \emptyset & \emptyset & \emptyset & \emptyset & \{\epsilon\} & \emptyset  \\
		d & \emptyset & \emptyset & \emptyset & \emptyset & \emptyset & \{\epsilon\} \\
		\end{block}
		\end{blockarray}
		%\end{tabular}\\
		\]

		It is easy to check that  $O^{out} = M^{out} \times N^{out}$ and $O^{in} =N^{in} \times M'^{in \times} N'^{out}$  are  given by the following matrices:
		\[	
		O^{out} =\begin{blockarray}{ccc}
		&  c & d \\
		\begin{block}{c(cc)}
		k & \{l_1l_3\} & \emptyset   \\
		l & \emptyset & \{l_2l_4\}  \\
		\end{block}
		\end{blockarray}
		\quad
		O^{in}=\begin{blockarray}{ccccc}
		& i & j & c & d \\
		\begin{block}{c(cccc)}
		a & \emptyset  & \emptyset & \{l_5\} & \emptyset  \\
		b & \{l_6l_7 \} & \{l_6l_8\} & \emptyset & \emptyset \\
		\end{block}
		\end{blockarray}
		\]
	
	The composition in this example enables us to zoom in on $Z$, in other words to know how $Z$ is built from $X$, itself a constituent box of $Y$. After the composition the link in Figure \ref{fig:wdboxxyz} (a representation of the two arrows $X \xrightarrow{\theta_1} Y$ and $Y \xrightarrow{\theta} Z$) from the port $h$ of Y to its port $f$,  vanishes in Figure \ref{fig:wdboxxz} (the arrow $X \xrightarrow{\theta  \circ \theta_1}  Z$) resulting from the composition. There are two reasons for this: the port $f$ is not connected in $X \xrightarrow{\theta_1} Y$, and the way $\cdot$ acts on links (or elements of $K$) in the semiring $(K,\cup,\cdot,\emptyset,\{\epsilon\})$.
	In practice this could mean that $f$ is not connected, even though it appears to be.
	
	\begin{proposition}
		\label{prop:identity}
		The identity law holds for the constituents of \msccatm.
	\end{proposition}
	\begin{proof}
		The matrices corresponding to $\theta \circ  1_b$ and  $1_a \circ \theta$ are seen to be equal to the matrix corresponding to $\theta$, whenever  $\theta:b \to a$ is an arrow.
	\end{proof}
	
	\begin{proposition}
		\label{prop:associativity}
		The associativity law holds for the constituents  of \msccatm.
	\end{proposition}
	\begin{proof}
		The associativity law follows from the associativity of matrix multiplication.
	\end{proof}
	
	\noindent Now we define the monoidal product of \msccatm.
	
	\begin{definition}
		\label{def:deftenpro}
		Let $X_1, X_2, Y_1, Y_2$ and $\theta_1 : X_1 \to Y_1$, $\theta_2 : X_2 \to Y_2$ be objects and arrows respectively of \msccatm. The monoidal product (or tensor product) $\otimes:\msccatm \times \msccatm \to \msccatm$ is given by: 
		\begin{equation}
		X_1 \otimes X_2 := (in(X_1)\cup in(X_2), out(X_1)\cup out(X_2)), \quad   \theta_1 \otimes \theta_2:=  \theta_1 \cup \theta_2
		\label{eq:tensorproduct} 
		\end{equation}
		
		The unit I is the box without input and output ports i.e. I:=$(\emptyset,\emptyset)$. R2 and R3 ensure that $\theta_1 \otimes \theta_2$ is well-defined.
	\end{definition}
	
	Note here that the effect of the tensor on different objects and arrows is simply to stack them straightforwardly and without losing (thanks to R1, R2 and R3) the origins of ports and links. For instance, stacking the unit I with any non-unit object $X$ will have no effect on $X$, since the box corresponding to I does not have any input or output ports. 
	The effect of tensor on the same object (resp. arrow) is: stacking a box (resp. an arrow) with itself has no effect, i.e., results in the same object (resp. arrow).
	
	\begin{proposition} $\otimes:\msccatm \times \msccatm \to \msccatm $ is a (bi)functor.
	\end{proposition}
	\begin{proof}
		$\otimes$ preserves identity morphisms and composition of morphisms. It suffices to write the corresponding morphisms, i.e. let $X_1, X_2, Y_1, Y_2, Z_1, Z_2$ and $\alpha_1 : Y_1 \to Z_1$, $\alpha_2 : Y_2 \to Z_2$, $\theta_1 : X_1 \to Y_1$, $\theta_2 : X_2 \to Y_2$ be objects and arrows respectively of \msccatm\
		\begin{itemize}
			\item[] Identity $\otimes(1_{X_1},1_{X_2})$ =  $1_{X_1} \cup 1_{X_2}$ = $X_1\cup X_2 \xrightarrow{1_{X_1\cup X_2 }}  X_1\cup X_2 = 1_{\otimes(X_1,X_2)}$
			\item[] Composition
			$\otimes(\alpha_1 \circ \theta_1,\alpha_2 \circ \theta_2)=  $ $\otimes( X_1 \xrightarrow{\alpha_1 \circ \theta_1}  Z_1, X_2 \xrightarrow{\alpha_2 \circ \theta_2} Z_2) =\\ 
			X_1 \cup X_2   \xrightarrow{(\alpha_1 \circ \theta_1)  \cup (\alpha_2 \circ \theta_2)}         Z_1 \cup Z_2$ and\\
			$\otimes(\alpha_1,\alpha_2) \circ \otimes(\theta_1,\theta_2) = $   
			$(Y_1\cup Y_2 \xrightarrow{\alpha_1 \cup \alpha_2}  Z_1\cup Z_2 )   \circ  (X_1\cup X_2 \xrightarrow{\theta_1 \cup \theta_2}   Y_1\cup Y_2)=\\
			X_1\cup X_2 \xrightarrow{(\alpha_1 \circ \theta_1) \cup (\alpha_2 \circ \theta_2)}  Z_1\cup Z_2$
		\end{itemize}
	\end{proof}

	\begin{proposition}
		The category (\msccatm,$\otimes$,I) with the constituents given in Definition \ref{def:defmscm} and the tensor product given in Definition \ref{def:deftenpro} is a symmetric monoidal category.
	\end{proposition}
	\begin{proof}
		We have already shown that \msccatm\:is a category (see Definition \ref{def:defmscm}, Proposition \ref{prop:associativity} and Proposition \ref{prop:associativity}). It remains to show that the tensor product is effectively a symmetric monoidal structure on \msccatm.
		
		Given $X_1,X_2,X_3 \in \msccatm_0$, we have
		\begin{equation}
		\begin{array}{rl}
		(X_1 \otimes X_2) \otimes X_3 =  X_1 \otimes (X_2 \otimes X_3) & (associativity)\\
		\text{I} \otimes X_1 = X_1 = X_1 \otimes \text{I} & (left\:\text{and} \: right \: units)\\
		X_1 \otimes X_2 = X_2 \otimes X_1 & (commutativity)
		\end{array}
		\label{eq:tensorproductproof} 
		\end{equation}
		The associator, unitors and inverse map (braiding) are identities.
	\end{proof}
	
	\begin{proposition}
		\msccat\:is an underlying symmetric multicategory of (\msccatm,$\otimes$).
		Its objects are the same objects as $\msccatm$, and its arrows are the set $\msccat(a_1,....,a_n;a)$. An arrow $a_1,....,a_n \to a$ is defined as an arrow $a_1 \otimes \dots \otimes a_n \to a$ in \msccatm.
	\end{proposition}
	\begin{proof}
		\emph{Any monoidal category ($A$,$\otimes$) has an underlying multicategory $C$} \cite[Example 2.1.3]{leinster2003higher}. Moreover,  \emph{any symmetric monoidal category is naturally a symmetric multicategory, via the symmetry maps	$\sigma  \cdot  {-} :  a_{\sigma(1)} \otimes \dots \otimes a_{\sigma(n)} \: \tilde{\longrightarrow} \: a_1 \otimes \dots \otimes a_n$}.\cite[Definition 2.2.21]{leinster2003higher}
	\end{proof}
	
	\textbf{Example} In Figure \ref{fig:wdbox}, there is an arrow $X_1, X_2 \to Y$ in \msccat, which is defined as $X_1 \otimes X_2 \to Y$ in (\msccatm,$\otimes$).
	
	Since we rely on matrix multiplication to define composition in \msccatm, we can redefine the components of \msccat\:using matrix without altering their initial meaning. This might be useful to automate computations with morphisms. Again a connection between morphisms of \msccatm\:and certain matrices over a semiring is presented in Appendix \ref{app:conmormat}.

	\begin{definition}
		\label{def:matmsc}
		\textbf{MatMsc} is an underlying symmetric multicategory of (\msccatm,$\otimes$). 
		
		$\bullet$ Its objects are the same objects as $\msccatm$. An object  $X$ is rewritten as  $(X^{in},X^{out})$.
		
		$\bullet$ Its arrows $\textbf{MatMsc}(a_1,....,a_n;a)$. An arrow $a_1,....,a_n \to a$ is rewritten as follows. Let $b = a_1 \otimes \dots \otimes a_n$,
		$X^{in} = in(b)$, $X^{out} = out(b)$, $Y^{in} = in(a)$ and  $Y^{out} = out(a)$,
		%Given two objects $(X^{in},X^{out})$ and $(Y^{in},Y^{out})$, 
		the arrow $b \to a$ in \msccatm, noted  $X \to Y$ is given by the couple of matrices
		$ (M^{in},M^{out})$ where $M^{in}$and $M^{out}$ are given by the functions $X^{in} \times (Y^{in} \cup X^{out}) \to K$ and 
		$Y^{out} \times X^{out} \to K$ respectively.
		Objects and arrows are subject to R1, R2 and R3 as in $\msccatm$.
		
		$\bullet$ Identities $\textbf{MatMsc}(a;a)$. For any object $(X^{in},X^{out})$, the identity arrow\\\hspace*{1cm} $1_X := (1_X^{in}:X^{in} \times (X^{in} \cup X^{out}) \to K, 1_X^{out}:X^{out} \times X^{out}  \to K)$, 
		%$1_X: X^{in} \times (X^{in} \cup X^{out}) \to K$ given by the couple  $(1_X^{in},1_X^{out})$
		rewritten as follows. $1_X^{in}$ yields
		\[
		\begin{blockarray}{rcc}
		& X^{in}  & X^{out} \\
		\begin{block}{c(cc)}
		X^{in}  & \mathbb{I}_M       &     \emptyset_M \\
		\end{block}
		\end{blockarray}
		\]
		and $1_X^{out}$ yields
		\[
		\begin{blockarray}{rc}
		& X^{out} \\
		\begin{block}{c(c)}
		X^{out}  & \mathbb{I}_M  \\
		\end{block}
		\end{blockarray}
		\]

		$\bullet$ The composition formula is the same as for $\msccatm$, but rewritten as follows. Consider the objects $X$, $Y$, $Z$ and arrows $n:X \to Y$ and $m:Y \to Z$ given by matrices $ (N^{in},N^{out})$ and $ (M^{in},M^{out})$ respectively. The arrow $m \circ n : X \to Z$ is given by $ (O^{in},O^{out})$: %(\ref{matrixarrowcomposition})
		\begin{equation*}
		\begin{array}{rl}
		O^{in} = & N^{in} \times M'^{in} \times N'^{out}\\
		O^{out} =& M^{out} \times N^{out}
		\end{array}
		%\label{eq:arrowinmsc4} 
		\end{equation*}
		% This ends the definition of \textbf{MatMsc}.
	\end{definition}

	The identity law holds in \textbf{MatMsc} with this rewriting (already shown in relation to \msccatm) i.e. $m \circ  1_X = m = 1_Y \circ m$ whenever  $m:X \to Y$ is an arrow. The matrices defining $m \circ  1_X$ and $1_Y \circ  m$ are both equal to $M$ i.e. $(M^{in},M^{out})$.
	The arrow $m \circ  1_X$ associates $M$ to $m$ and $N$ to $1_X$. Therefore $O^{in}$ equals $N^{in} \times M'^{in} \times N'^{out} =$
	\[
	\begin{blockarray}{rll}
	&X^{in} & X^{out}  \\
	\begin{block}{r(cc)}
	X^{in} & \mathbb{I}_M       &     \emptyset_M \\
	\end{block}
	\end{blockarray}
	\quad \times
	\begin{blockarray}{rcc}
	& Y^{in}\cup X^{out}  & X^{out} \\
	\begin{block}{c(cc)}
	X^{in}  & M^{in}        &     \emptyset_M \\
	X^{out} & \emptyset_M   &  \mathbb{I}_M  \\
	\end{block}
	\end{blockarray}
	\quad \times
	\begin{blockarray}{rcc}
	& Y^{in} & X^{out} \\
	\begin{block}{c(cc)}
	Y^{in}  &  \mathbb{I}_M &     \emptyset_M \\
	X^{out} & \emptyset_M   &  \mathbb{I}_M   \\
	X^{out} & \emptyset_M   &     \mathbb{I}_M \\
	\end{block}
	\end{blockarray} 
	\]
	
	$O^{in}$ is indeed equal to $M^{in}$. We also have  $O^{out}$ equals $M^{out} \times N^{out} =$
	\[
	%\begin{tabular}{c c c}
	%	$M^{out}$ & $N^{out}$ & $N^{in}$ \\
	\begin{blockarray}{rl}
	&X^{out}  \\
	\begin{block}{r(c)}
	Y^{out} & M^{out} \\
	\end{block}
	\end{blockarray}
	\quad \times
	\begin{blockarray}{rl}
	&X^{out}  \\
	\begin{block}{r(c)}
	X^{out} & \mathbb{I}_M \\
	\end{block}
	\end{blockarray}
	\]
	$O^{out}$ is equal to $M^{out}$. Hence, $m \circ  1_X = m$. In the same way: $m= 1_Y \circ  m$.
	
	\textbf{MatMsc} can be seen as the category where objects declare interfaces and where morphisms declare how composite interfaces are internally built from other interfaces. The composition formula lets us safely zoom in and zoom out on levels further down and up in the hierarchy of interfaces.
	
	It is worth noting that with only \textbf{MatMsc} (or \textbf{Msc} or \textbf{Mscm}), no meaning (related to a domain of interest or business domain) is assumed for an interface, its ports and, if applicable, its internal links.

	\subsection{Interpretations of the structure}
	\label{sec:specs}
	The present section proposes an approach that seeks to bridge the gap between a possible semantics (interpretation of the structure) and its implementation in practice.
	A semantics corresponding to an interpretation of the structure may be formally defined in a variety of ways (denotational, axiomatic, operational, etc.) and with numerous possibilities of implementation. This paper is not concerned with defining or implementing a particular semantics, but sets out principles for formally connecting the structure, semantics and their implementation. As mentioned above, the architecture of a system is not about the effective (or final) implementation of the system.

	Using the structure of models, it is possible to specify an interpretation of boxes such that the associated semantics can be recovered. An actual interpretation could therefore be independently developed by using the structure of models. Ideally, the specification of the interpretation is given by a functor $F:\msccat \to \msicat$, $\msicat$ being the category in which the structure is interpreted. The caveat 'ideally' is necessary, because it may not always be straightforward to define a functorial semantics.
	
	\textbf{Examples}
		
		(i) Consider once again a model corresponding to the box shown in Figure \ref{fig:wdbox}. A domain-specific meaning of this model can only be known once the definition of an interpretation of its structure is given. Let us suppose that the domain of interest is classical physics and that we wish to study some physical properties of boxes. In this case, $F$ can be defined as $F: \msccat \to \msicat$. Let $\msicat$ :=\textbf{Set}. Along with objects, set of real numbers, and arrows, $n$-ary functions define the relationship between a physical property of a box and the physical properties of its $n$ constituent boxes. $F$ is a map of multicategories and an algebra for \msccat\:(See \cite[Definition 2.1.9, Definition 2.1.12]{leinster2003higher}). Here, let us take \textit{mass} to be the physical property under consideration.
		
		Let $F$ be  $F_0:\{X_1,X_2,Y\}\to \mathbb{R}_{> 0}$ given by $X \mapsto ma$. The image of $X_1,X_2 \to Y$ is the function $F_0(X_1) \times F_0(X_2)\to F_0(Y)$ given by $(ma_1,ma_2) \mapsto ma$.
		
		Depending on the meaning of links, $ma$ can be equal either to (*) $ma_1+ma_2$  or to (**) $ma_1+ma_2+m_L$, $m_L \in \mathbb{R}_{+}$ .
		In the first case (*), the mass of a composite box will depend solely on the masses of its constituent boxes. In the second case (**), the mass of a box will additionally depend on $m_L$, which is an additional mass emerging from internal connections.

		Similar interpretations can be made for other physical properties.\\
		\begin{figure}[!h]
			\centering
			\includegraphics[scale=0.25]{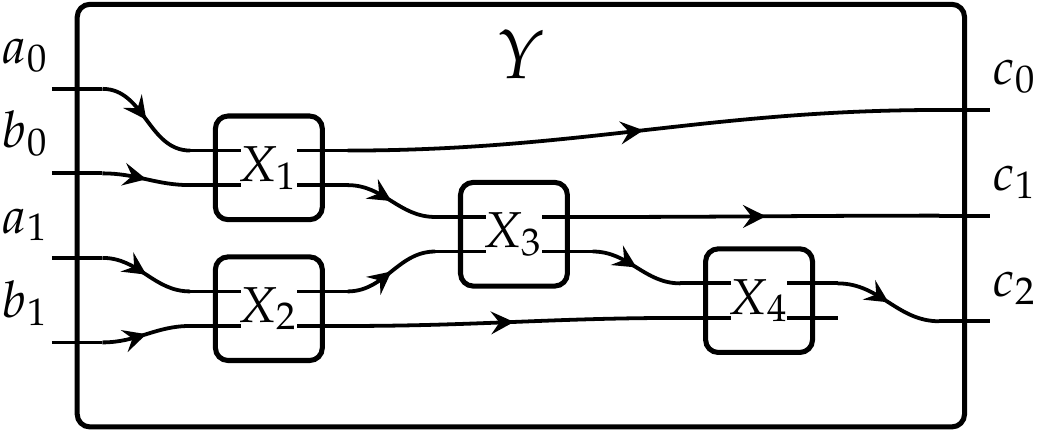}
			\caption{$X_1, X_2, X_3, X_4 \to Y$ \label{fig:add2}}
		\end{figure}
		
		(ii) Consider the arrow in Figure \ref{fig:add2}. Let each of the boxes $X_i, i=1..4$ be interpreted as a \textit{two-digit adder} ($d_0 + d_1=cs$) which returns the sum ($s$) and the carry ($c$) via the top and bottom output ports respectively, given $d_0$ and $d_1$ obtained from top and bottom input ports. Here, the containing box $Y$ can be interpreted as \textit{two-digit numbers adder} ($a_1a_0 + b_1b_0=c_2c_1c_0$). There is a category (\msicat) with five objects and one arrow.
	
	The concept of mass in (i), and the concepts of two-digit adder and two-digit numbers adder in (ii), must be exposed by  $F$, irrespective of how these concepts are defined in \msicat\:and irrespective of the implementation.
	These data must be machine-accessible.
	It is worth noting that the functorial semantics of wiring diagrams (i.e. special cases of \msccat) defined in \cite{rupel2013operad}, \cite{spivak2013operad}, \cite{vagner2015algebras} and \cite{spivak2015nesting} are concerned with the interconnections of discrete-time processes, databases and plug-and-play circuits, differential equations of open dynamical systems, and mode-dependent networks respectively. In each case, a particular algebra is devised which may be implemented as definitions of interpretations of the structure of models.

	\subsection{Usages of the structure}
	\label{sec:structureusage}
	If the structure of a model is to be automatically exchangeable, then in the same way that the definition of an interpretation of the structure needs to be machine-accessible, the related data must also be machine-accessible.
	
	The data relating to an actual (or candidate) system modelled is defined by a Set-valued functor $I:\msccat \to \textbf{Set}$. This links an object $a \in \msccat_0$ to the set of its usages $I_0(a)$, and any arrow $a_1, ... ,a_n \xrightarrow{\theta} a \in C(a_1, . . . ,a_n;a)$ to the function  $I_0(a_1)\times ... \times I_0(a_n) \xrightarrow{I(\theta)} I_0(a)$. Taking the composition and the identity in the same way as in $\textbf{Set}$, $I$ is a map of multicategories \cite[Definition 2.1.9]{leinster2003higher} and an algebra for \msccat\: \cite[Definition 2.1.12]{leinster2003higher}. The usefulness of $I$ stems from the fact that where the structure (\msccat) of a model, together with at least one interpretation ($F$) of this structure, has been established, this is still not quite enough to allow the model to be used. The actual business data corresponding to the components of the structure are also required.  For instance, consider the box shown in Figure \ref{fig:wdbox} and the interpretation given in example (i) above. How do we communicate the fact that the masses of boxes $X_1$ and $X_2$ and $Y$ are given by actual values? No assumptions are made concerning the relations between the different masses. A parallel can be drawn with databases: it is good to have a database schema and to know what tables correspond to in the real world, but in order to use a database we also need to know the contents of the tables (the actual data entries). The purpose of $I$ is to answer this kind of question, by providing instances of the components of a structure of a model.
	
	Let us recapitulate the purpose of the present section. We have provided an informal and a formal understanding of the proposed architecture FA. To be able to read a model (data), it is necessary to have at least the specification of $I:\msccat \to \textbf{Set}$, while to be able to utilize a model (semantics), it is necessary to have at least the specification of $F: \msccat \to \msicat$.\\
	$I$ and $F$ are intended to facilitate the transfer of the  modularity of $\msccat\:$ in different interpretations of $\msccat$. Without functoriality, the modularity and the generality offered by the category-theoretic framework are lost (even though this functoriality is not always easily perceived, and may in some cases be unnecessary).
	
	For exploiting models through a given automated routine, it is necessary to specify
	\begin{equation}
	%\vspace{-0.2cm}
	\msicat \xleftarrow{F} \msccat \xrightarrow{I} \textbf{Set}
	%\vspace{-0.3cm}
	\label{eq:foundmodels}
	\end{equation} and make the definition of interpretation machine-accessible. Nonetheless, wherever the usage of a structure is linked to a particular interpretation of that structure, it can be argued that the effective data related to $\msccat$ should instead be given by an amalgamated sum (a pushout) over (\ref{eq:foundmodels}).

\section{Related research}
\label{sec:relatedresearchchap4}
FMI\footnote{\label{fn:fmiwebsite} https://www.fmi-standard.org/} (Functional Mockup Interface), the standard which attempts to solve the problem stated in Section \ref{sec:problemchap4}, is discussed in Section \ref{sec:r1}. Then other relevant related works that can be correlated to FA are discussed in Section \ref{sec:r2}.

\subsection{FMI}
\label{sec:r1}
FMI is a tool-independent interface that allows models to be exchanged and inter-operated in a co-simulation.
An FMU (Functional Mock-up Unit), one of the components of FMI, must implement the FMI API via XML files and C-code. In FMI a model is associated with a particular FMU, which is generated using different tools.
The kinds of models covered by FMI are (co) simulation models \cite{blochwitz2012functional}. This paper is not primarily concerned with how models co-operate. Nevertheless, the way a model is defined (such as via an FMU) will influence how it can be reused.

An FMU separates the description of interface data (via an XML -- Extensible Markup Language -- file) from its mode of operation (via a program written in the C programming language).
The interface data describes the static elements of the model such as type definitions (Real, Integer etc.), model variables,  model structure (an ordered list of outputs termed \textit{Outputs}, states termed \textit{Derivatives}), and some initially unknown data termed \textit{InitialUnknowns}), along with units etc.\\
The C-code and header files define the functions to be implemented by an FMU. These functions are used to simulate a model described according to a particular FMU. Generally speaking, the functions are for purposes such as initializing the FMU, assigning values to an input variable, retrieving the value of an output variable, changing the state of the FMU, etc.
More information about the evolving specification of the FMI standard and supported tools can be found on the FMI website.

Although the FMI standard provides an interface to be implemented by FMUs, one concern is that there can be semantic gaps (studied in \cite{tripakis2015bridging} \cite{cremona2016fide}) between different model semantics (e.g, discrete events, dataflow models) or modelling languages and the target interface  \cite{tripakis2016compositionality}.  The semantics of the target interface has been considered as a timed Mealy machine \cite{tripakis2015bridging}.
Another concern is the extent to which the interface can be used to capture the semantics of different kinds of models \cite{tripakis2016compositionality}. Attempts to address this problem have resulted in various (proposed) adaptations of FMI \cite{cremona2016fide} \cite{cremona2019hybrid}. Therefore, the specification and definition of the common interface are crucial. The FMI standard is mainly concerned with the exchange of models and their simulation as black boxes.

In this paper we adhere to the FMI standard in separating structure and function. But, unlike in FMI, we give a formal definition of FA. The structures of models are considered as composite components defined within a category-theoretic framework, and a predefined interface for the implementation of functions is not imposed. Any implementation has to be derived from an interpretation of the structure of models. The (software) code corresponding to an implementation must expose the functions whose inputs and outputs are correlated to an interpretation of the structure. The data corresponding to the instantiations (i.e. actual serialized models) of the structure are elements that should be persisted as a file.

\subsection{Other works}
\label{sec:r2}
The kind of descriptive model that we are talking about is close, in terms of model structure (or abstract syntax), to the component-based, modular, and even composite-like models that are habitually encountered in system and software modelling, simulation, and more generally, analysis. See for instance \cite{hardebolle2008modhel} \cite{tripakis2013modular} for approaches and tools that deal with this kind of model structure. Although these approaches address modelling and simulation, their ultimate aim is not to solve the problem that is our concern here. They generally introduce their own approach (mainly on modelling formalisms) to system modelling and analysis.
Metamodels for this kind of structure have been presented by \cite{lee2010disciplined} \cite{hardebolle2008modhel}. Their approaches advocate hierarchical and component-based design in systems modelling and engineering.

From a practical (yet formal) point of view, a box (a component of FA) is analogous to an \textit{actor} in \cite{tripakis2013modular}. In \cite{tripakis2013modular}, boxes are encoded using the logic of the modelling tool Ptolemy and the Java programming language. This may hamper the reuse of models, albeit some of these models may be imported using FMI.

In the present paper, unlike in \cite{tripakis2013modular}, a box does not correspond to a particular semantics (i.e. the target of an interpretation of the structure). For \cite{tripakis2013modular}, boxes or actors are seen as extended timed state machines, since the authors' focus is on behavioural models. The corresponding semantics is also called the actor interface, defined in a way that includes both the structure and the semantics of the box. However, the actor interface can also be considered as the target of an interpretation of the structure of a box, since the input and output ports of a box are explicit in this interface. In \cite{tripakis2013modular} composite actors are basically considered as a set of actors, whereas for us a composite box is an arrow in \msccat\: (see Section \ref{sec:applicationsct}).
It will be noted that \cite{tripakis2013modular} is mainly concerned with a unified description of the behavioural semantics of a composite that comprises more than one actor.

It might be suggested that graphs are a way of formalizing the structure of models. Although graphs are very useful in representing interconnected objects, they lack the sort of intrinsic feature for dealing with the hierarchical composite structure of vertices (in the present context vertices equate to boxes) that we require. In order to represent this kind of structure, the graph would need to include additional data, which might complicate the formalization.

In fact, hierarchical graphs have also been formalized in category theory to model hierarchical structures \cite{Milner08} (where coupling of hyper-graphs and forests yields bigraphs) \cite{BruniCGLM11} (graphs inside graphs). Although hierarchical graphs encapsulate the hierarchy of nodes (here, boxes) and the communication links between nodes, the nature and structure of nodes and node containment (nesting) and links between nodes do not have the same meaning as described in Section \ref{sec:modelorganisationstructure}. In the present work a node is given by input and output ports. The source and target of a link are ports at the same level of hierarchy of a box. We do not use two separate structures to define boxes (objects) and as a basis for composite boxes (arrows), and in this our approach contrasts with the two-structure approaches of \cite{Milner08} (place graph and linking graph ) and \cite{BruniCGLM11} (graph layer and hierarchical graph).

Our inspiration for the category structure considered in this paper is, rather, to be found in \cite{spivak2013operad}  \cite{rupel2013operad}, \cite{vagner2015algebras} and \cite{luzeaux2015formal}. 
\cite{luzeaux2015formal} considers symmetric monoidal and compact closed categories, with associated logics and proof theories as a possible formal foundation for systems engineering, given that it is compatible with current modelling languages and tools and encompasses existing formalisms.
The mathematical object presented in this paper for encapsulating the structure of models, a symmetric multicategory (see Section \ref{sec:applicationsct}), is closely related to the Wiring Diagrams (WD) structure developed in a series of papers: \cite{rupel2013operad}, \cite{spivak2013operad},  \cite{vagner2015algebras} etc. The two big differences between the WD structure and our structure are the following:

(i) Our structure is a more general construction supported by the "wiring" or connection pattern, i.e. almost all possible ways of connecting sub-components (or constituents) of a composite component are allowed, unlike WD (see Section \ref{sec:modelstructure} below for details of what WD does not allow). A WD is basically given by
\begin{equation}
%\vspace{-0.2cm}
(X^{in} + X^{out}) \xrightarrow{} (Y^{in} + X^{out})\xleftarrow{} (Y^{in} + Y^{out})
%\vspace{-0.3cm}
\label{eq:wd}
\end{equation} 
\begin{comment}or, more precisely
\begin{equation}
%\vspace{-0.2cm}
(X^{in} + Y^{out}) \xrightarrow{} (Y^{in} + X^{out})
%\vspace{-0.3cm}
\label{eq:wd2}
\end{equation}
\end{comment}
which can be decomposed into two functions
\begin{equation}
\begin{array}{rl}
\phi^{in} :& X^{in} \to  Y^{in} + X^{out}\\
\phi^{out} :& Y^{out} \to  X^{out}
\end{array}
%\vspace{-0.1cm}
\label{eq:wdfunctions} 
\end{equation}

To obtain (\ref{eq:arrowinmsc1}) -- our definition -- from (\ref{eq:wdfunctions}), the elements of $\phi^{in}$ and $\phi^{out}$ must be associated to the abstract sets of links or wires $L^{in}$ and $L^{out}$. To get (\ref{eq:wdfunctions}) from (\ref{eq:arrowinmsc1}) it is necessary to have $\theta^{in}$ and $\theta^{out}$ of (\ref{eq:arrowinmsc1}) that are injective. The elements of $L^{in}$ and $L^{out}$ can be seen as the labels of wires.

(ii) We are not concerned with a particular interpretation (e.g. an algebra) of the structure of models.

	\section{Summary and perspective}
	\label{sec:conclusionchap4}
	Our aim was to provide a federated architecture for descriptive models in order to facilitate their development, sharing and utilization. Though clear and precise definitions of models are essential, how such definitions are conveyed in practice, particularly in distributed, collaborative and computer-aided environments, is also essential, and crucial to some extent. In order to bridge the gap between theoretical and practical considerations, a federated architecture (FA) was presented via the specification: $\msicat \xleftarrow{F} \msccat \xrightarrow{I} \textbf{Set}$ (\ref{eq:foundmodels}).
	Categorical settings are useful for abstracting away from details, and for differentiating and correlating structures of models and their usages. Since there are well-known basic data structures, algorithms and software libraries related to matrices, their use may foster calculations with FA.
	Setting up this framework can make it possible to communicate and utilize different descriptive models, transparently or independently of the tool used to build them and without sacrificing the definitions of models.
	
	One question not addressed in this paper is the way in which models change and evolve over time. This aspect might be manageable canonically using the structure of models. It would be possible to include within a structure an `off' object, the `black hole' (to model a destroyed box) yielding a new structure $K$. A functor from a category $K_t$ at time $t$ to the same category $K_{t'}$ at time $t'> t$ could then be defined. By taking the time instants as objects in a category $T$ featuring some order, and defining a functor from $T$ to $K$, we would then be in a position to address change and evolution in models with the structure K\cite{ehresmann1987hierarchical}.

\begin{comment}

\section*{Acknowledgment}
This work is carried out and funded in the framework of the Labex MS2T. It is supported by the French Government, through the program "Investments for the future" managed by the National Agency for Research (Reference ANR-11-IDEX-0004-02).

This work is co-funded by the Airbus Defence and Space company.
\end{comment}

\bibliographystyle{eptcs}

%\label{lastpage}
%\appendix
\section{Appendix}
\subsection*{Connection between morphisms of \msccatm\:of and certain matrices over a semiring}
\label{app:conmormat}
	The connection assumed in Definition \ref{def:defmscm} is established and proved in three steps as follows. (1) an equivalence between certain functions and matrices over a semiring is proved. Then (2) certain functions are shown to be equivalent to the morphisms of \msccatm, by rewriting the signatures of morphisms' component functions ($f^{in}$ and $f^{out}$). Finally, it is deduced (3) that a morphism of \msccatm\: yields a matrix being one of these certain matrices and the converse.  Computations relating to morphisms are then made via matrices.
	Finally, matrix multiplication is used to obtain the composition formula of \msccatm.

	Let $K$ be the power set of the set of strings of finite lengths. $(K,\cup,\cdot,\emptyset,\{\epsilon\})$ is a semiring \cite{Droste2009} where $\times := \cdot $ is the product induced by the string concatenation operator, $+: = \cup$ is the addition given by union of sets of strings of finite lengths, $0 := \emptyset$ is the zero given by the empty set, and $1 := \{\epsilon\}$ is the unit given by the singleton set containing the empty string.  Let also $I$ and $J$ be the finite sets with $n$ and $m$ elements respectively.

	\begin{proof}[Proof of (1)]
		Let $f : I \times J \to K$  be a function. Then $M(n,m)$ is a matrix, given by $M=a_{ij}, a_{ij}\in K, i=1..n, j=1..m$. Conversely, given a matrix $M=a_{ij}, a_{ij}\in K, i=1..n, j=1..m$, the corresponding function is $f:I \times J \to K$ given by $(i,j) \mapsto  a_{ij}$. 
		This completes the proof of (1).
	\end{proof}

	\begin{proof}[Proof of (2)]
		The data (\ref{eq:arrowinmsc1}) of an arrow can be rewritten as follows.
		We define two functions $f^{out}$, $f^{in}$
		\begin{equation}
		\begin{array}{rrrl}
		f^{out} :& Y^{in} \times (Y^{out} \cup Z^{in})& \to& 2^{L^{in}}\\
		f^{in} :& Z^{out} \times Y^{out}& \to& 2^{L^{out}}
		\end{array}
		%\vspace{-0.1cm}
		\label{eq:arrowinmsc4} 
		\end{equation}\\
		where $Y^{in} = in(b)$, $Y^{out} = out(b)$,  $Z^{in} = in(a)$ and $Z^{out} = out(a)$. $2^L$ is the power set.
		Like the functions $\theta^{in}$ and $\theta^{out}$ that are total functions, $f^{out}$ and $f^{in}$ are also total functions given by
		associating with each element $\alpha$ of their domain the set of elements that are each the preimage of $\alpha$  in $\theta^{in}$ and $\theta^{out}$ respectively.
		From (\ref{eq:arrowinmsc1}), we have $2^{L^{in}}, 2^{L^{out}}\subset K$.
		Hence, the functions $\theta^{in}$ and $\theta^{out}$ are equivalent via ($f^{out}$ and $f^{in}$) to certain functions $f: I \times J \to K$.
		%Conversely, suppose given, by construction a function $f: P \times Q \to K$.  $f$ already has the form of $f^{in}$. Let $P:=Y^{in}$, if we demand that $Q = Y^{out}\cup Z^{in}$
		%we obtain a function defined by \ref{eq:arrowinmsc4}.
		This proves (2).
	\end{proof}
	
	\begin{proof}[Proof of (3)]
		(3) follows from (1) and (2).
	\end{proof}
	
	Henceforth, matrix multiplication is used for the composition formula of \msccatm.
	
	\textbf{Example} Let us consider $f_A : I_1 \times S\to K$, and   $f_B : S \times I_2 \to K$ where $I_1=\{a,b\}$, $S = \{c,d\}$ and $I_2=\{e,f\}$.\\
		Suppose given \\
		$f_A := \{(a,c) \mapsto \{w_1\}, (a,d) \mapsto \emptyset, (b,c) \mapsto \{w_2,w_3\},(b,d) \mapsto \emptyset\}$\\
		$f_B := \{(c,e) \mapsto \{w_5\}, (c,f) \mapsto \emptyset, (d,e) \mapsto \emptyset,(d,f) \mapsto \emptyset\}$
		
		The matrices $A$ and $B$ corresponding to the functions $f_A$ and $f_B$ are as follows.
		\[
		%\begin{tabular}{c c c}
		%	$M^{out}$ & $N^{out}$ & $N^{in}$ \\
		A=\begin{blockarray}{rcc}
		& c & d \\
		\begin{block}{c(cc)}
		a & \{w_1\} & \emptyset  \\
		b & \{w_2,w_3\} & \emptyset \\
		\end{block}
		\end{blockarray}
		\quad
		%&
		B=\begin{blockarray}{rcc}
		& e & f \\
		\begin{block}{c(cc)}
		c & \{w_5\} & \emptyset  \\
		d & \emptyset & \emptyset  \\
		\end{block}
		\end{blockarray}
		\]
	
	The matrix $A$ means in practice that there is a link or wire ($w_1$) coming from the port $c$ and supplying the port $a$. There are also two links ($w_2$ and $w_3$) coming from $c$ and supplying $b$. The product $A \times B$ given by
	\[
	A \times B =\begin{blockarray}{rcc}
	& e & f \\
	\begin{block}{c(cc)}
	a & \{w_1w_5\} & \emptyset  \\
	b & \{w_2w_5,w_3w_5\} & \emptyset \\
	\end{block}
	\end{blockarray}
	\]
	shows in practice that there is also a link $(w_1w_5$) coming from $e$ and supplying $a$. Typically, a port will be either an input or output port of a box.
\end{document}